\documentclass[11pt,letterpaper]{article}
\usepackage[margin=1in]{geometry}
\usepackage[utf8]{inputenc}
\usepackage{booktabs} 

\usepackage[linesnumbered, vlined, ruled]{algorithm2e} 
\SetAlFnt{\small}
\SetAlCapFnt{\small}
\SetAlCapNameFnt{\small}
\SetAlCapHSkip{0pt}
\IncMargin{-\parindent}

\newcommand{\opt}{\ensuremath{\operatorname{OPT}}}

\newcommand{\expec}{\ensuremath{\mathbb{E}}}
\newcommand{\dist}{\ensuremath{\mathcal{D}}}
\newcommand{\principal}{\ensuremath{\mathcal{P}}}
\newcommand{\princ}{\principal}
\newcommand{\agent}{\ensuremath{\mathcal{A}}}

\usepackage{tikz}
\usepackage{float}
\usepackage{amsmath}
\usepackage{amssymb}
\usepackage{amsthm}
\usepackage{graphicx}
\usepackage[hidelinks]{hyperref}

\newtheorem{lemma}{Lemma}
\newtheorem{theorem}{Theorem}
\newtheorem{proposition}{Proposition}

\newtheorem{corollary}{Corollary}
\newtheorem{example}{Example}
\newtheorem{remark}{Remark}

\newcommand{\av}{\ensuremath{a}}
\newcommand{\pv}{\ensuremath{b}}
\newcommand{\pr}{\ensuremath{p}}

\title{Delegated Online Search}

\author{Pirmin Braun \thanks{Goethe University Frankfurt, Germany, Email: pirminbraun16@gmail.com.}
    \and Niklas Hahn \thanks{Goethe University Frankfurt, Germany, Email: nhahn@em.uni-frankfurt.de. Hahn gratefully acknowledges the support of GIF grant I-1419-118.4/2017.}
    \and Martin Hoefer \thanks{Goethe University Frankfurt, Germany, Email: mhoefer@em.uni-frankfurt.de. Hoefer gratefully acknowledges the support of GIF grant I-1419-118.4/2017 and DFG grants Ho 3831/5-1, 6-1, and 7-1.}
    \and Conrad Schecker \thanks{Goethe University Frankfurt, Germany, Email: schecker@em.uni-frankfurt.de.}
}
\date{}

\begin{document}
        \maketitle

    \begin{abstract}
        In a delegation problem, a \emph{principal} $\princ$ with commitment power tries to pick one out of $n$ options. Each option is drawn independently from a known distribution. Instead of inspecting the options herself, $\princ$ delegates the information acquisition to a rational and self-interested \emph{agent} $\agent$. After inspection, $\agent$ proposes one of the options, and $\princ$ can accept or reject.

        Delegation is a classic setting in economic information design with many prominent applications, but the computational problems are only poorly understood.
        In this paper, we study a natural \emph{online} variant of delegation, in which the agent searches through the options in an online fashion. For each option, he has to irrevocably decide if he wants to propose the current option or discard it, before seeing information on the next option(s). How can we design algorithms for $\princ$ that approximate the utility of her best option in hindsight?

        We show that in general $\princ$ can obtain a $\Theta(1/n)$-approximation and extend this result to ratios of $\Theta(k/n)$ in case (1) $\agent$ has a lookahead of $k$ rounds, or (2) $\agent$ can propose up to $k$ different options. We provide fine-grained bounds independent of $n$ based on two parameters. If the ratio of maximum and minimum utility for $\agent$ is bounded by a factor $\alpha$, we obtain an $\Omega(\log\log \alpha / \log \alpha)$-approximation algorithm, and we show that this is best possible. Additionally, if $\princ$ cannot distinguish options with the same value for herself, we show that ratios polynomial in $1/\alpha$ cannot be avoided. If the utilities of $\princ$ and $\agent$ for each option are related by a factor $\beta$, we obtain an $\Omega(1 / \log \beta)$-approximation, where $O(\log \log \beta / \log \beta)$ is best possible.

    \end{abstract}




	\section{Introduction}

	The study of delegation problems is a prominent area in economics with numerous applications. There are two parties -- a decision maker (called \emph{principal}) $\princ$ and an \emph{agent} $\agent$. $n$ actions or \emph{options} are available to $\princ$. Each option has a utility for $\princ$ and a (possibly different) utility for $\agent$, which are drawn from a known distribution $\dist$. Instead of inspecting options herself, $\princ$ delegates the search for a good option to $\agent$. $\agent$ sees all realized utility values and sends a signal to $\princ$. Based on this signal (and $\dist$), $\princ$ chooses an option. Both parties play this game in order to maximize their respective utility from the chosen option. 
	
	Many interesting applications can be captured within this framework. For example, consider a company that is trying to hire an expert in a critical area. Instead of searching the market, the company delegates the search to a head-hunting agency that searches the market for suitable candidates. Alternatively, consider an investor, who hires a financial consultant to seek out suitable investment opportunities. Clearly, principal and agent might not always have aligned preferences. While the investor might prefer investments with high interest rates, the financial consultant prefers selling the products for which he gets a provision.

	In applications such as searching for job candidates or financial investments, availability of options often changes over time, and the pair of agents needs to solve a stopping problem. For example, many lucrative financial investment opportunities arise only within short notice and expire quickly. Therefore, a consultant has to decide whether or not to recommend an investment without exactly knowing what future investment options might become available. Hence, $\agent$ faces an online search problem, in which the $n$ options are realized in a sequential fashion. After seeing the realization of option $i$, he has to decide whether to propose the option to $\princ$ or discard it. If the option is proposed, $\princ$ decides to accept or reject this option and the process ends. Otherwise, the process continues with option $i+1$.

	In the study of delegation problems, $\princ$ usually has commitment power, i.e., $\princ$ specifies in advance her decision for each possible signal, taking into account the subsequent best response of $\agent$. This is reasonable in many applications (e.g., an investor can initially restrict the investment options she is interested in, or the company fixes in advance the required qualifications for the new employee). Interestingly, although $\princ$ commits and restricts herself in advance, this behavior is usually in her favor. The induced best response of $\agent$ can lead to much better utility for $\princ$ than in any equilibrium, where both parties mutually best respond. Using a revelation-principle style argument, the communication between $\princ$ and $\agent$ can be reduced to $\agent$ revealing the utilities of a single option and $\princ$ deciding to accept or reject that option (for a discussion, see, e.g.~\cite{KleinbergK18}).

	The combination of online search and delegation has been examined before, albeit from a purely technical angle. Kleinberg and Kleinberg~\cite{KleinbergK18} recently designed approximation algorithms for delegation, using which $\princ$ can obtain a constant-factor approximation to the expected utility of her best option in hindsight. Their algorithms heavily rely on techniques and tools developed in the domain of prophet inequalities. However, they are applied to an \emph{offline} delegation problem. Instead, our model is an extension of~\cite{KleinbergK18} to online search. Interestingly, our results reveal a notable contrast -- in online delegation a constant-factor approximation might be impossible to achieve. In fact, in the worst case the approximation ratio can be as low as $\Theta(1/n)$ and this is tight. Motivated by this sharp contrast, we provide a fine-grained analysis based on two natural problem parameters: (1) the discrepancy of utility for the agent, and (2) the misalignment of agent and principal utilities. 

\subsection{Model}

We study \emph{online delegation} between principal $\princ$ and agent $\agent$ in (up to) $n$ rounds. In every round $i$, an option is drawn independently from a known distribution $\dist_i$ with finite support $\Omega_i$ of size $s_i$. We denote the options of $\dist_i$ by $\Omega_i = \{\omega_{i1}, \ldots, \omega_{i,s_i}\}$ and the random variable of the draw from $\dist_i$ by $O_i$. For every $i \in [n]$ and $j \in [s_i]$, the option $\omega_{ij}$ has probability $\pr_{ij}$ to be drawn from $\dist_i$. If this option is proposed by $\agent$ and chosen by $\princ$, it yields utility $\av_{ij} \ge 0$ for $\agent$ and $\pv_{ij} \ge 0$ for $\princ$.

We assume that $\princ$ has commitment power. Before the start of the game, she commits to an \emph{action scheme} $\varphi$ with a value $\varphi_{ij} \in [0,1]$ for each option $\omega_{ij}$. $\varphi_{ij}$ is the probability that $\princ$ accepts option $\omega_{ij}$ when it is proposed by $\agent$ in round $i$. We will sometimes consider \emph{deterministic} action schemes, which we represent using sets $E_i = \{ \omega_{ij} \mid \varphi_{ij} = 1\}$ of \emph{acceptable options} in each round $i \in [n]$.

In contrast to $\princ$, $\agent$ gets to see the $n$ random draws from the distributions in an online fashion. He has to decide after each round whether he proposes the current option $O_i$ to $\princ$ or not. If he decides to propose it, then $\princ$ decides according to $\varphi$ whether or not she accepts the option. If $\princ$ accepts, the respective utility values are realized; if not, both players receive utility 0. In either case, the game ends after $\princ$ decides. Clearly, both players strive to maximize their expected utility.

Initially, both players know the distribution $\dist_i$ for every round $i \in [n]$. The sequence of actions then is as follows: (1) $\princ$ decides $\varphi$ and communicates this to $\agent$; (2) in each round $i$, $\agent$ sees $O_i \sim \dist_i$ and irrevocably decides to propose or discard $O_i$; (3) as soon as $\agent$ decides to propose some option $O_i = \omega_{ij}$, then $\princ$ accepts it with probability $\varphi_{ij}$, and the game ends. 

Because $\agent$ knows the distributions and the action scheme $\phi$ of upcoming rounds which modifies his expected utility from proposed options, $\agent$ essentially faces an online stopping problem which he can solve using backwards induction. Hence, we can assume without loss of generality that all decisions (not) to propose an option by $\agent$ are deterministic. That is, if the expected utility from the realization in the current round is greater than the expected utility obtained in the upcoming round, propose the current option, otherwise, wait for the next round.
To avoid technicalities in the analysis, we assume that $\agent$ breaks ties in favor of $\princ$.


Our goal is to design action schemes $\varphi$ with high expected utility for $\princ$. We compare the expected utility to the one in the non-delegated (online) optimization problem, where $\princ$ searches through the $n$ realized options herself. The latter is an elementary stopping problem, for which a classic prophet inequality relates the expected utility of the optimal online and offline strategies by at most a factor of 2~\cite{KrengelS77,KrengelS78}.

We also analyze scenarios with \emph{oblivious} and \emph{semi-oblivious proposals}. In both these scenarios, $\agent$ reveals only the utility value $\pv_{ij}$ for $\princ$ when proposing an option (but not his own value $\av_{ij}$). In contrast, when $\princ$ gets to see the utility values of both agents, we term this \emph{conscious proposals}. The difference between semi-oblivious and (fully) oblivious scenarios lies in the prior knowledge of $\princ$. In the semi-oblivious scenario, $\princ$ is fully aware of the distributions, including all potential utility values $\av_{ij}$ for $\agent$. In the oblivious scenario, $\princ$ initially observes the probabilities of all options along with her utility values $\pv_{ij}$, but the values $\av_{ij}$ of $\agent$ remain unknown to $\princ$ throughout. In the scenarios with restricted discrepancy studied in Section~\ref{sec:max-min-agent-ratio}, $\princ$ is aware of the bound $\alpha = \max_{i,j} \av_{ij} / \min_{i,j} \av_{ij}$.




\begin{example}
Consider the following simple example for illustration. We consider deterministic strategies by $\princ$ and conscious proposals. There are two rounds with the options distributed according to Table~\ref{table:example}.
\begin{table}
    \centering

\begin{tabular}{c||cc|cc}
    round $i$ & \multicolumn{2}{c|}{1} & \multicolumn{2}{c}{2} \\ \hline
    option $\omega_{ij}$ & $\omega_{11}$ & $\omega_{12}$ & $\omega_{21}$ & $\omega_{22}$ \\
    value-pair $(a_{ij},b_{ij})$ & (3,1)  & (3,8) & (2,4) & (16,4) \\
    probability $p_{ij}$ & 0.75 & 0.25 & 0.75 & 0.25 \\
\end{tabular}

\caption{An example instance}
\label{table:example}
\end{table}
%
%
For the benchmark, we assume that $\princ$ can see and choose the options herself. The best option is $\omega_{12}$. If this is not realized in round 1, the option realized in round 2 is the best choice. Note that this optimal choice for $\princ$ can be executed even in an online scenario, where she first sees round 1 and gets to see round 2 only after deciding about round 1. The expected utility of this best (online) choice for $\princ$ is $5$.

Now in the delegated scenario, suppose $\princ$ accepts all option $\omega_{22}$. Then $\agent$ would always wait for round 2 and hope for a realization of $\omega_{22}$, even if $\omega_{21}$ would not be accepted by $\princ$. Hence, accepting $\omega_{22}$ leads to an expected utility for $\princ$ of at most $4$.
In contrast, the optimal decision scheme for $\princ$ is to accept only $\omega_{12}$ and $\omega_{21}$ with an expected utility of $4.25$.

Note that for the (semi-)oblivious scenario, $\princ$ cannot distinguish the options in round $2$ which decreases her expected utility to at most 4.

This shows that $\princ$ has to strike a careful balance between (1) accepting a sufficient number of high-profit options to obtain a high expected utility overall and (2) rejecting options to motivate $\agent$ to propose options that are better for $\princ$ in earlier rounds. \hfill $\blacklozenge$
\end{example}

\subsection{Contribution}
    In Section~\ref{sec:lb} we show that the worst-case approximation ratio for online delegation is $\Theta(1/n)$ and this is tight. Intuitively, $\agent$ waits too long and forgoes many profitable options for $\princ$. $\princ$ can only force $\agent$ to propose earlier if she refuses to accept later options -- this, however, also hurts the utility of $\princ$. The instances require a ratio of maximum and minimum utility values for $\agent$ that is in the order of $n^{\Theta(n)}$.
    We further show that this lower-bound instance can be used for extensions in which (1) $\agent$ has a lookahead of $k$ rounds, or (2) $\agent$ can propose up to $k$ options, resulting in tight approximation ratios of $\Theta(k/n)$.

    In the Section~\ref{sec:max-min-agent-ratio}, we examine the effect of the discrepancy of utility for $\agent$ using the ratio $\alpha$ of maximum and minimum utility values. We obtain an $\Omega(\log\log \alpha/\log \alpha)$-approximation of the optimal (online) search for $\princ$, which is tight. The algorithm limits the acceptable options of $\princ$, partitions them into different bins, and then restricts $\agent$'s search space to the best possible bin for $\princ$. The challenge is to carefully design a profitable set of options that should be accepted by $\princ$ without giving $\agent$ an incentive to forgo proposing many of these options. Our algorithm shows that even if differences in utility of $\agent$ are polynomial in $n$, a good approximation for $\princ$ can be obtained.

    Additionally, we consider more challenging \emph{semi-oblivious} and \emph{oblivious} scenarios in which $\princ$ does not get to see the agent's utility of the proposed option. In the (fully) oblivious case, $\princ$ is even apriori unaware of the utility values for $\agent$ for all options (and thus remains so throughout). In the semi-oblivious case, $\princ$ knows the prior distributions fully, i.e., for every option the probability and the utility values for \emph{both} agents.

    Our Algorithms~\ref{algo:alpha-approx} and~\ref{algo:semiOblivious-approx} achieve $\Omega(1/\alpha)$ and $\Omega(1/(\sqrt{\alpha}\log \alpha))$-approximations for oblivious and semi-oblivious scenarios, respectively. This is contrasted with a set of instances for which any action scheme cannot extract more than an $O(1/\alpha)$- and $O(1/\sqrt{\alpha})$-approximation in the oblivious and semi-oblivious scenarios, respectively. These results highlight the effect of the hiding of $\agent$'s utilities from $\princ$ (in the proposal, or in the proposal and the prior) -- the achievable approximation ratios increase from logarithmic to polynomial ratios in $\alpha$.

    In Section~\ref{sec:principal-agent-ratio}, we consider the misalignment of agent and principal utilities via a parameter $\beta \ge 1$, which is the smallest value such that all utilities of $\princ$ and $\agent$ are related by a factor in $[1/\beta, \beta]$. Limited misalignment also leads to improved approximation results for $\princ$. We show an $\Omega(1/\log \beta)$-approximation of the optimal (online) search for $\princ$. Moreover, every algorithm must have a ratio in $O(\log \log \beta / \log \beta)$.  For the agent-oblivious variant, we obtain an $\Omega(1/\beta)$-approximation, whereas every algorithm must have a ratio in $O(1/\sqrt{\beta})$.



\bigskip
\subsection{Related Work}
Holmstrom \cite{Holmstrom77,Holmstrom84} initiated the study of delegation as a bilevel optimi\-za\-tion between an uninformed principal and a privately informed agent. The principal delegates the decision to the agent who himself has an interest in the choice of decision. Since the principal has the power to limit the search space, her optimization problem lies in striking the balance between restricting the space enough such that the second-level optimization by the agent doesn't hurt her too much and allowing a large enough set of potential decisions such that an acceptable decision can be found at all. Holmstrom identified sufficient conditions for a solution to the problem to exist. Subsequent papers \cite{MelumadS91,AlonsoM08} studied the impact of (mis-)alignment of the agent's and the principal's interests on the optimal delegation sets.

In another direction, the model was extended by allowing the principal to set costs for subsets instead of forbidding them \cite{AmadorB13,AmbrusE17}. These costs might be non-monetary, i.e., using different levels of bureaucracy for different subsets of options.

Armstrong and Vickers~\cite{ArmstrongV10} studied the delegation problem over discrete sets of random cardinality with elements drawn from some distribution. They identify sufficient conditions for the search problem to have an optimal solution. A similar model was considered by Kleinberg and Kleinberg~\cite{KleinbergK18}, where the option set searched by the agent consists of $n$ iid draws from a known distribution. Their results include constant-factor approximations of the optimal expected principal utility when performing the search herself rather than delegating it to the agent. For this, they employ tools from online stopping theory. The key difference between their work and our paper is that -- although using tools from online optimization -- they study an \emph{offline} problem while we focus on an \emph{online} version.

Bechtel and Dughmi~\cite{BechtelD21} recently extended this line of research by combining delegation with stochastic probing. Here a subset of elements can be observed by the agent (subject to some constraints), and several options can be chosen (subject to a different set of constraints). They provide constant-factor approximations for several downwards-closed constraint systems.

The study of persuasion, another model of strategic communication, has gained a lot of traction at the intersection between economics and computation in recent years. Here, the informed party (the ``sender'') is the one with commitment power, trying to influence the behavior of the uninformed agent (the ``receiver''). Closely related to our paper is the study of persuasion in the context of stopping problems~\cite{HahnHS20,HahnHS20IJCAI}. These works study persuasion problems in a prophet inequality~\cite{HahnHS20IJCAI} as well as in a secretary setting~\cite{HahnHS20}. 

Other notable algorithmic results on persuasion problems concern optimal policies, hardness, and approximation algorithms in the general case~\cite{DughmiX16} as well as in different variations, e.g., with multiple receivers~\cite{BabichenkoB17,BadanidiyuruBX18,DughmiX17,Rubinstein17,Xu20}, with limited communication complexity~\cite{DughmiKQ16,GradwohlHHS21}, or augmenting the communication through payments~\cite{DughmiNPW19}. A more elaborate communication model with evidence was studied recently in the framework of persuasion as well as delegation~\cite{HoeferMP21}.


\section{Impossibility}\label{sec:lb}

\subsection{A Tight Bound}

As a first simple observation, note that $\princ$ can always achieve an $n$-approximation with a deterministic action scheme, even in the agent-oblivious case. $\princ$ accepts exactly all options in a single round $i^*$ with optimal expected utility, i.e., $E_{i^*} = \{ \omega_{i^*,j} \mid j \in [s_{i^*}] \}$ for $i^* = \arg \max_{i \in [n]} \expec[\pv_{ij}]$, and $E_j = \emptyset$ otherwise. This motivates $\agent$ to always propose the option from round $i^*$, and $\princ$ gets expected utility $\expec[\pv_{i^*,j}]$. By a union bound, the optimal utility from searching through all options herself is upper bounded by $\expec\left[ \sum_i \pv_{ij} \right] \le n \cdot \expec[\pv_{i^*,j}]$.

\begin{proposition}
    For online delegation there is a deterministic action scheme $\varphi$ such that $\princ$ obtains at least a $1/n$-approximation of the expected utility for optimal (online) search.
\end{proposition}

We show a matching impossibility result, even in the IID setting with $\dist_i = \dist$ for all rounds $i \in [n]$, and when $\princ$ gets to see the full utility pair of any proposed option. There are instances in which $\princ$ suffers a deterioration in the order of $\Theta(n)$ over the expected utility achieved by searching through the options herself.

For the proof, consider the following class of instances. The distribution $\dist$ can be cast as an independent composition, i.e., we independently draw the utility values for $\princ$ and $\agent$. For $\princ$ there are two possibilities, either utility $1$ with probability $1/n$, or utility 0 with probability $1-1/n$. For $\agent$, there are $n$ possibilities with agent utility of $n^{4\ell}$, for $\ell = 1,\ldots,n$, where each one has probability $1/n$. In combination, we can view $\dist$ as a distribution over $j=1,\ldots,2n$ options. Options $\omega_j$ for $j=1,\ldots,n$ have probability $1/n^2$ and utilities $(\pv_{j}, \av_{j}) = (1,n^{4j})$, for $j =n+1,\ldots,2n$ they have probability $1/n-1/n^2$ and utilities $(\pv_{j}, \av_{j}) = (0,n^{4(j-n)})$.

\begin{theorem}\label{thm:generalLB}
	There is a class of instances of online delegation in the IID setting, in which every action scheme $\varphi$ obtains at most an $O(1/n)$-approximation of the expected utility for optimal (online) search.
\end{theorem}

\begin{proof}
For simplicity, we first show the result for schemes $\varphi$ with $\varphi_{ij} = 0$ for all rounds $i \in [n]$ and all $j = n+1, \ldots, 2n$. In the end of the proof we discuss why this can be assumed for an optimal scheme.

Since all options $j \in [n]$ have the same utility for $\princ$, she wants to accept one of them as soon as it appears. If she searches through the options herself, the probability that there is an option of value 1 is  $1-(1-1/n)^n \ge 1-1/e$. Her expected utility is a constant. In contrast, when delegating the search to $\agent$, the drastic utility increase motivates him to wait for the latest round in which a better option is still acceptable by $\princ$. As a result, $\agent$ waits too long, and removing acceptable options in later rounds does not remedy this problem for $\princ$.

More formally, interpret an optimal scheme $\varphi$ as an $n \times n$ matrix, for rounds $i\in [n]$ and options $j \in [n]$. We outline some adjustments that preserve the optimality of matrix $\varphi$.

Consider the set $S$ of all entries with $\varphi_{ij} \le 1/n$. For each $(i,j) \in S$, the probability that option $j$ is realized in round $i$ is $1/n^2$. When it gets proposed by $\agent$, then it is accepted by $\princ$ with probability at most $1/n$. By a union bound, the utility obtained from all these options is at most $1 \cdot |S| \cdot 1/n^2 \cdot 1/n \le 1/n$.

Suppose we change the scheme by decreasing $\varphi_{ij}$ to 0 for each $(i,j) \in S$. Then each entry in $\varphi$ is either 0 or at least $1/n$. If $\agent$ makes the same proposals as before, the change decreases the utility of $\princ$ by at most $1/n$. Then again, in the new scheme $\agent$ can have an incentive to propose other options in earlier rounds. Since all options with $\varphi_{ij} \neq 0$ have utility 1 for $\princ$, this only leads to an increase of utility for $\princ$. Moreover, in round 1 we increase all acceptance probabilities to $\varphi_{1j} = 1$ for $j \in [n]$. Then, upon arrival or such an option $\varphi_{1j}$, the change can incentivize $\agent$ to propose this option -- which is clearly optimal for $\princ$, since this is an optimal option for her. Since the change is in round 1, it introduces no incentive to wait for $\agent$. As such, it can only increase the utility for $\princ$.

Now consider any entry $\varphi_{ij} \ge 1/n$. We observe two properties:
\begin{enumerate}
	\item Suppose $\varphi_{i'j'} \ge 1/n$ for $i' < i$ and $j' < j$. Then $\princ$ accepts realization $\omega_{j'}$ in round $i'$ with positive probability, but she will also accept the better (for $\agent$) realization $\omega_{j}$ in a later round $i$ with positive probability. $\agent$ will not propose $\omega_{j'}$ in round $i'$ but wait for round $i$, since the expected utility in the later round $i$ is at least $n^{4j} \cdot 1/n^2 \cdot \varphi_{ij} \ge n^{4j-3} > n^{4(j-1)} \ge n^{4j'} \cdot \varphi_{i'j'}$, the utility in round $i'$. As such, we assume w.l.o.g.\ that $\varphi_{i'j'} = 0$ for all $i' < i$ and $j' < j$.
	\item Suppose $\varphi_{i'j} < \varphi_{ij}$ for $i' < i$. By property 1., all realizations $\omega_{j'}$ with $j' < j$ are not accepted in rounds $1,\ldots,i-1$. Hence, setting $\varphi_{i'j} = \varphi_{ij}$ does not change the incentives for $\agent$ w.r.t.\ other options, and thus only (weakly) increases the expected utility of $\princ$. By the same arguments, we set $\varphi_{ij'} = \max\{\varphi_{ij'}, \varphi_{ij}\}$ for all inferior options $j' < j$ in the same round $i$.
\end{enumerate}

We apply the previous two observations repeatedly, starting for the entries $\varphi_{in}$ in the $n$-th column for option $\omega_n$, then in column $n-1$, etc. By 1., every positive entry $\varphi_{ij} \ge 1/n$ leads to entries of 0 in all ``dominated'' entries $\varphi_{i'j'}$ with $i' < i$ and $j' < j$. As a consequence, the remaining positive entries form a ``Pareto curve'' in the matrix $\varphi$ or, more precisely, a Manhattan path starting at $\varphi_{1n}$, ending at $\varphi_{n1}$, where for each $\varphi_{ij} \ge 1/n$ the path continues either at $\varphi_{i+1,j} \ge 1/n$ or $\varphi_{i,j-1} \ge 1/n$.

We can upper bound the expected utility of $\princ$ by assuming that all $2n-1$ entries on the Manhattan path are 1 (i.e., $\varphi$ is deterministic) and $\agent$ proposes an acceptable option whenever possible. The probability that this happens is at most $(2n-1)/n^2 = O(1/n)$ by a union bound. This is an upper bound on the expected utility of $\princ$ and proves the theorem for schemes with $\varphi_{ij} = 0$ for all $i \in [n]$ and $j \ge n+1$.

Finally, suppose $\varphi_{ij} > 0$ for some $j \ge n+1$. Clearly, option $\omega_{j}$ adds no value to the expected utility of $\princ$. Moreover, the fact that $\omega_{j}$ has positive probability to be accepted in round $i$ can only motivate $\agent$ to refrain from proposing inferior options in earlier rounds. As such, setting $\varphi_{ij} = 0$ only (weakly) increases the utility of $\princ$.
\end{proof}

The lower bound in Theorem~\ref{thm:generalLB} remains robust also in several extensions of the model.

\subsection{Extensions}
\label{sec:extensions}
We discuss two generalizations for which a slight adaptation of the above lower bound is sufficient. First, we consider the case that $\agent$ has a lookahead. Second, we allow several proposals to be made by $\agent$.
\subsubsection{Agent with Lookahead}

Consider online delegation when $\agent$ has a lookahead of $k$ rounds. In round $i$, $\agent$ gets to see all realized options of rounds $i,i+1,\ldots,\min\{n,i+k\}$. For simplicity, our benchmark here is the expected value of $\princ$ for optimal (non-delegated) \emph{offline} search (i.e., online search with lookahead $n-1$). Note that the expected value for optimal online search is at least 1/2 of this~\cite{KrengelS77,KrengelS78}. Hence, asymptotically all benchmarks of expected utility for optimal offline or online search, with or without lookahead, are the same.

\begin{proposition}
    \label{prop:lookahead}
    For online delegation with lookahead $k$ there is an action scheme $\varphi$ such that $\princ$ obtains an $\Omega(k/n)$-approximation of the expected utility for optimal (offline) search.
\end{proposition}

Partition the $n$ rounds into $\lceil n/(k+1) \rceil $ parts with at most $k+1$ consecutive rounds each. Suppose we apply (non-delegated) offline search on each part individually. The expected value of offline search on the best of the $O(n/k)$ parts yields an $\Omega(k/n)$-approximation of the expected value of offline search on all $n$ rounds.

To obtain an $\Omega(k/n)$-approximation for the online delegated version, apply online search with $\agent$ and lookahead of $k$ to the best part of at most $k+1$ consecutive rounds. Due to the lookahead, this results in offline delegated search. In terms of utility for $\princ$, offline delegated search using prophet-inequality techniques~\cite{KleinbergK18} approximates optimal offline search by at least a factor of 1/2. Hence, applying the offline delegation algorithm of~\cite{KleinbergK18} on the best set of $k+1$ consecutive rounds yields an $\Omega(k/n)$-approximation.

Let us show that this guarantee is asymptotically optimal. The argument largely follows the proof of Theorem~\ref{thm:generalLB}. The class of instances is the same. We only explain which parts of the proof must be adapted.

\begin{corollary}\label{thm:lookaheadLB}
    There is a class of instances of online delegation with lookahead $k$ in the IID setting, in which every action scheme $\varphi$ obtains at most an $O(k/n)$-approximation of the expected utility for optimal (offline) search.
\end{corollary}

\begin{proof}
    Using similar observations as in the proof of Theorem~\ref{thm:generalLB} above, we can again (a) assume w.l.o.g.\ that $\varphi_{ij} = 0$ for all $j=n+1,\ldots,2n$, and (b) assume that $\varphi_{ij} = 0$ or $\varphi_{ij} \ge 1/n$, for all $i=1,\ldots,n$ and $j=1,\ldots,2n$, which deteriorates the expected utility for $\princ$ by at most $O(1/n)$.

    Consider the two properties in the proof of Theorem~\ref{thm:generalLB}. For property (1), we extend the idea to entries $\varphi_{i'j'} \ge 1/n$ with $j' < j$ and $i'+ k < i$. In particular, $\agent$ will decide to wait and not propose option $\omega_{j'}$ in round $i'$ if there is a round $i > i'+ k$ with a better option $\omega_j$ being acceptable (with probability at least $\varphi_{ij} \ge 1/n$). As such, we drop $\varphi_{i'j'}$ to 0 whenever such a constellation arises. Then, whenever an entry remains $\varphi_{i'j'} \ge 1/n$, this means that all entries for better options $j > j'$ in rounds $i = i'+k+1,\ldots,n$ must be $\varphi_{ij} = 0$.

    Now for a given option $\omega_{j'}$, consider round $i_{j'} = \arg\min \{i \mid \varphi_{ij'} \ge 1/n\}$. Then, for all better options with $j > j'$, property (1) requires that $\varphi_{ij} = 0$ for all $i \ge i_{j'} + k+1$. As such, for each option $\omega_j$, there can be at most $k$ positive entries ``beyond the Manhattan path'', i.e., $k$ rounds in which $\omega_j$ remains acceptable (with prob.\ at least $1/n$) after the first round when any lower-valued $\omega_{j'}$ becomes acceptable (with prob.\ at least $1/n$).

    Property (2) applies similarly as before. As such, we obtain a Manhattan path with $2n-1$ entries, and in addition there can be up to $nk$ entries with $\varphi_{ij} \ge 1/n$, i.e., a total of at most $(k+2)n - 1$ entries. We again upper bound the expected utility of $\princ$ by assuming that all these entries are 1 and $\agent$ proposes an acceptable option whenever possible. The probability that this happens is at most $((k+2)n-1)/n^2 = O(k/n)$ by a union bound, and this implies the upper bound on the expected utility of $\princ$.
\end{proof}

\subsubsection{Agent with $k$ Proposals}

Now consider the case when $\agent$ can propose up to $k$ options in $k$ different rounds. In this case, the definition of an action scheme becomes more complex -- rather than a single matrix, $\varphi$ turns into a decision tree. For each round $i$, consider the \emph{history} $H_i = (h_1,\ldots,h_{i-1})$. For every round $j = 1,\ldots,i-1$, the entry $h_j$ indicates whether or not there was a proposal by $\agent$ in round $j$, and if so, which option was proposed. Now given a round $i$ and a history $H_i$, an action scheme yields a value $\varphi_{ij}(H_i) \in \{0,1\}$ indicating whether or not $\princ$ accepts option $\omega_j$ when being proposed in round $i$ after history $H_i$. As before, $\princ$ commits to an action scheme anticipating the induced rational behavior of $\agent$. A simple backward induction shows that there is always an optimal proposal strategy for $\agent$ is deterministic. For simplicity, we also restrict attention to deterministic action schemes for $\princ$.

\begin{proposition}
    For online delegation with $k$ proposals there is a deterministic action scheme $\varphi$ such that $\princ$ obtains an $\Omega(k/n)$-approximation of the expected utility for optimal (offline) search.
\end{proposition}

The scheme is related to the approach in the previous section. Select the best interval $\ell, \ldots, \ell+k-1$ of $k$ consecutive rounds that maximize the expected value of offline search for $\princ$ over these rounds. We observed in the previous section that optimal offline search in these $k$ rounds yields an $\Omega(k/n)$-approximation of optimal offline search over $n$ rounds. We design an action scheme that incentivizes $\agent$ to propose exactly the $k$ options in rounds $\ell, \ldots, \ell+k-1$, thereby reducing the scenario to  (non-delegated) online search for $\princ$ over these rounds. Since the performance of online and offline search are related by a factor of 2, asymptotically we achieve the same performance as offline search over the $k$ rounds.

We set $\varphi_{ij}(H_i) = 0$ for rounds $i < \ell$ and all $j$ and $H_i$, as well as for rounds $i > \ell + k -1$ and all $j$ and $H_i$. For each round $\ell \le i \le \ell+k-1$, we set $\varphi_{ij}(H_i) = 0$ for all options $j$ whenever the history reveals that in at least one of the rounds $\ell,\ldots,i-1$ there was no proposal from $\agent$. Otherwise, if $H_i$ reveals that there was a proposal in each of these rounds, we set $\varphi_{ij}(H_i)$ as in an optimal online (non-delegated) search over rounds $\ell,\ldots,\ell+k-1$.

In this action scheme, $\princ$ immediately terminates the search whenever she did not receive a proposal from $\agent$ in one of the $k$ rounds, leaving both agents with a utility of 0. This creates an incentive for $\agent$ to submit a proposal in each of the $k$ rounds, since this is the only possibility to obtain a positive utility value.

Let us show that the approximation guarantee is asymptotically optimal. The argument uses and extends the result of Theorem~\ref{thm:generalLB}. The class of instances is the same.

\begin{theorem}\label{thm:kProposalLB}
    There is a class of instances of online delegation with $k$ proposals in the IID setting, in which every deterministic action scheme $\varphi$ obtains at most an $O(k/n)$-approximation of the expected utility for optimal (offline) search.
\end{theorem}

\begin{proof}
    We analyze the process by splitting the evolution of the process into at most $k$ non-overlapping \emph{phases}. Let $i_{\ell}$ be the (random) round in which $\agent$ makes the $\ell$-th proposal, for $\ell = 1,\ldots,k$. For completeness, we set $i_0 = 0$. Phase $\ell$ is the set of rounds $\{i_{\ell-1} + 1,\ldots, i_{\ell}\}$. $\princ$ can accept an option in at most one of the phases.
    Thus, by lineary of expectation, the expected utility of $\princ$ is upper bounded by the sum of expected utilities that $\princ$ obtains in each phase. In the rest of the proof, we will show that in each phase, the expected utility for $\princ$ is at most $O(1/n)$. Hence, the total expected utility of $\princ$ is $O(k/n)$, which proves the approximation guarantee.

    Towards this end, consider a single phase $\ell$. We condition on the \emph{full history} of the process before phase $\ell$, i.e., we fix all options drawn as well as decisions of $\princ$ and $\agent$ that have led to the $(\ell-1)$-th proposal in round $i_{\ell-1}$.
    We denote this full history by $H^f$. During phase $\ell$ (conditioned on $H^f$), we want to interpret the process as a single-proposal scenario analyzed in Theorem~\ref{thm:generalLB}. In particular, by fixing the history and the starting round of phase $\ell$, the histories $H_i$ within phase $\ell$ are also fixed. As such, during phase $\ell$, the scheme $\varphi$ can be seen as an action scheme for a single-proposal scenario with $n-i_{\ell-1}$ rounds.

    Now let us define an auxiliary single-proposal instance with $n$ rounds. In this instance, we assume $\princ$ sets $\varphi'_{ij} = 0$ for all options $j=1,\ldots,2n$ in the first $i=1,\ldots,i_{\ell-1}$ rounds and then in rounds $i = i_{\ell-1}+1,\ldots,n$ applies $\varphi'_{ij} = \varphi_{ij}(H_i)$ (where $H_i$ is composed of $H^f$ and no proposal in rounds $i_{\ell-1}+1,\ldots,i-1$).
    Then $\princ$ behaves in the auxiliary instance exactly as in phase $\ell$ of the $k$-proposal instance. In contrast, $\agent$ does not necessarily show the same behavior. In the auxiliary instance, $\agent$ gets utility 0 if the proposal is rejected.
    In phase $\ell$ of the $k$-proposal instance, however, proposing an option that gets rejected can be profitable for $\agent$. After rejection, phase $\ell + 1$ is reached and better expected utility for $\agent$ might be achievable in upcoming rounds (since the scheme could result in more favorable behavior of $\princ$ when the $\ell$-th reject decision happens in round $i$).

    In the auxiliary instance, we model this property by a \emph{reject bonus} for $\agent$ -- whenever a proposal is rejected in any round $i \ge i_{\ell-1}+1$,
    then ($\princ$ receives no utility and) $\agent$ receives the conditional expected utility from optimal play in the remaining rounds of the $k$-proposal instance, conditioned on $\varphi$ and history $H_{i+1}$ composed of $H_{i}$ and the rejected proposal in round $i$. It is straightforward to see that in the auxiliary instance with reject bonus, the interaction between $\princ$ and $\agent$ exactly proceeds as in phase $\ell$ of the $k$-proposal instance.

    Consider the auxiliary single-proposal instance with reject bonus for any given phase $\ell$.
    We prove that the expected utility for $\princ$ does not decrease when the following \textit{adjustments} are made: (1) the reject bonus is reduced to 0, and (2) we set $\varphi'_{ij} = 0$ for all rounds $i$ and options $j \ge n+1$ (i.e., the ones with $b_{j} = 0$). We prove the statement by induction over the rounds.

    Clearly w.l.o.g.\ there are no proposals in rounds $i \le i_{\ell-1}$. We can assume that (1) and (2) hold for all these rounds. Now consider round $i = i_{\ell-1} + 1$. When rejecting an option, or when accepting an option of value 0, the utility for $\princ$ is 0. For these options, adjustments (1) and (2) in round $i$ change the utility for $\agent$ to 0, as well. When facing such an option in round $i$, the adjustments incentivize $\agent$ to wait for potentially better subsequent options. It weakly increases the expected utility of $\princ$.

    Towards an induction, suppose the statement is true after the adjustments (1) and (2) in all rounds $i_{\ell-1} + 1 \le i' \le i$. Now consider round $i+1$. First, condition on the event that in both instances (with and without adjustments on round $i+1$), we reach round $i+1$. As argued above, the adjustments in round $i+1$ imply that $\agent$ has less incentive to propose an option of value 0 for $\princ$ in round $i+1$ and more incentive to wait for subsequent rounds. Hence, the utility for $\princ$ (conditioned on reaching round $i+1$) does not decrease.

    Note, however, that the probability of reaching round $i+1$ also changes by the adjustments. For every $i_{\ell-1} < i' < i+1$, removing the reject bonus and reducing the set of acceptable options in round $i+1$ lead to a reduction in expected utility for $\agent$ from the rounds \emph{after} round $i'$. This increases the probability that $\agent$ will propose an option in some round $i'$ \emph{before} $i+1$. It decreases the probability of reaching round $i+1$. Nevertheless, this is again good news for $\princ$: Since by hypothesis $\princ$ accepts only options of utility 1 and there is no reject bonus  in all rounds $i' \le i$, any proposal in these rounds is accepted and results in optimal utility for $\princ$. Overall, $\princ$ only profits from the adjustments (1) and (2) in round $i+1$. By induction, this holds when the adjustments are made in all rounds.

    After the adjustments, the auxiliary instance is a standard single-proposal instance studied in the context of Theorem~\ref{thm:generalLB}. This shows that the expected utility obtained by $\princ$ is in $O(1/n)$.

    As a consequence, the conditional expected utility for $\princ$ in phase $\ell$ (conditioned on each $H^f$) is upper bounded by $O(1/n)$. Hence, the overall expected utility from phase $\ell$ is at most $O(1/n)$. The expected utility from $k$ phases is upper bounded by $O(k/n)$. This proves the theorem.
\end{proof}

\section{Discrepancy of Agent Utilities}\label{sec:max-min-agent-ratio}

\subsection{Conscious Proposals}

The lower bound instances in Theorem~\ref{thm:generalLB} rely on an exponentially large ratio of agent utilities between 1 and $n^{O(n)}$. Is such a drastic discrepancy necessary to show a substantial lower bound? Can we obtain better approximation results for instances with a smaller ratio of the maximum and minimum utility values for $\agent$?

We first assume that $\av_{ij} > 0$ for all options (see Remark~\ref{rem:loglogWithZero} below how to extend our results to the case when $a_{ij} = 0$ is allowed). Let $\alpha = \max\{ \av_{ij} \mid i \in [n], j \in [s_i] \} / \min\{ \av_{ij} \mid i \in [n], j \in [s_i] \}$. W.l.o.g.\ we scale all utility values to $\av_{ij} \in [1,\alpha]$, where both boundaries $\alpha$ and $1$ are attained by at least one option. We say that the agent has \emph{$\alpha$-bounded utilities}.



\begin{algorithm}[t]\DontPrintSemicolon
	\caption{$\Omega(\log \log \alpha / \log \alpha)$-Approximation}\label{algo:log-alpha-loglog-alpha-approx}
	\KwIn{$n$ distributions $\dist_1, \dots, \dist_n$}
	\KwOut{Action Scheme $\varphi$}\vspace{0.3cm}
    Let $Q = $ RestrictOptions$(\dist_1, \dots, \dist_n,1/2)$.\;
    \lIf{$Q$ spans only a single round}{
        Set $B_1 = Q$.}
    \Else{
	Construct $c = \lceil \log_2 \alpha \rceil$ classes $\mathcal{C}_1, \dots, \mathcal{C}_c$ such that $\mathcal{C}_k = \{(i,j) \in Q \mid \av_{ij} \in [2^{k-1},2^k)\}$ for all $k = 1, \dots, c - 1$ and $\mathcal{C}_{c} =  \{ (i,j) \in Q \mid a_{ij} \in [2^{c-1},2^c]\}$. \label{algo-loglog-log:constructClasses}  \\
	Set $b=1$, and $s = c$. Open bin 1 and set $B_1 = \emptyset$. \\
	\For{$k = c$ \textnormal{down to} $1$}{\label{algo-loglog-log:startFor}
		\If{$2^{k-1} < 2^s \cdot \sum_{(i,j) \in B_b \cup \mathcal{C}_k} p_{ij}$ }{ \label{step10}
			set $b = b+1$ and $s=k$. \tcp*{$\sum_{(i,j) \in Q} p_{ij} < 1/2$, so no open bin stays empty}
			Open the new bin $b$ and set $B_b = \emptyset$}
		Add class $\mathcal{C}_k$ to bin $B_b =  B_b \cup \mathcal{C}_k$.}
	}\label{algo-loglog-log:endFor}
	$b^* = \arg \max_{b=1,\ldots} \sum_{(i,j) \in B_b} p_{ij} \pv_{ij}$, the bin with highest utility for $\princ$.\\
	Set $\varphi_{ij} = 1$ for all $(i,j) \in B_{b^*}$ and $\varphi_{ij} = 0$ otherwise.\\
	\textbf{return} $\varphi$
\end{algorithm}

We use Algorithm~\ref{algo:log-alpha-loglog-alpha-approx} to compute a good action scheme.
Intuitively, we partition the best options for $\princ$ that add up a total probability mass of roughly $1/2$ into $O(\log \alpha / \log \log \alpha)$ many bins. Each bin is constructed in a way such that $\agent$ is incentivized to propose the first option he encounters from that particular bin. The algorithm determines the best bin for $\princ$ and deterministically restricts the acceptable options to the ones from this bin.

The algorithm uses the subroutine RestrictOptions$(\dist_1,\dots,\dist_n,m)$ (Algorithm~\ref{algo:semiOblivious-restrict}), which considers all possible options in descending order of principal value until a total mass of $m$ is reached. These options are then collected in the set $Q = \{(i,j) \mid b_{ij} \ge b_{i'j'} \ \forall (i',j') \notin Q\}$ such that $\sum_{(i,j) \in Q} p_{ij} < m$. The options in the set $Q$ are compared with the set consisting solely of the first option to surpass the combined mass of $m$. Whichever gives the higher expected utility for principal is then returned by the subroutine as $Q$. This ensures that $\sum_{(i,j) \in Q \cup B_0} p_{ij} \cdot \pv_{ij} \ge m/2 \cdot \expec_{\omega_{ij} \sim \dist_i}[\max_{i \in [n]} \pv_{ij}]$. We formally prove this in Lemma~\ref{lem:restrictOptions} below.

\begin{algorithm}[t]\DontPrintSemicolon
    \caption{\label{algo:semiOblivious-restrict} RestrictOptions}
    \KwIn{$n$ distributions $\dist_1,\dots, \dist_n$, value $m$ restricting the mass}
    \KwOut{Set $Q$ of good options for $\princ$}
    Set $Q = \emptyset$, $p=p^*=0$, $U = \bigcup_{i=1}^n \bigcup_{j=1}^{s_i} \{ (i,j) \}$.\\
    \While{$p < m$}{\label{restrictOptions:startWhileQB}
        $U^* = \emptyset$, $p^* = 0$ \\
        \For{$k = 1, \dots, n$}{\label{restrictOptions:startForSamePrincValue}
            Let $U_k^* = \{ (k,j) \in U \mid \pv_{kj} \ge \pv_{i'j'}$ for all $(i',j') \in U \}$ be the options in round $k$ from the set of all remaining options with the best utility for $\princ$\\
            Set $p_k = \sum_{(i,j) \in U^*_k} p_{ij}$\\
            \lIf{\label{restrictOptions:startIfCombinedMass}$p^* + p_k < m$}{add $U^* = U^* \cup U_k^*$, update $p^* = p^* + p_k$}
            \uElse{ \lIf{$p_k > p^*$}{ set $U^* = U_k^*$}
                \textbf{break for-loop}}\label{restrictOptions:endIfCombinedMass}

        }\label{restrictOptions:endForSamePrincValue}
        Set $p^* = \sum_{(i,j) \in U^*} p_{ij}$ and $\pv^* = \pv_{ij}$ for $(i,j) \in U^*$\\
        \tcc*{Note that all options in $U^*$ have the same value for $\princ$.}
        \textbf{if} $p + p^* > m$ \textbf{then} set $B = U^*$; \textbf{else} add $Q = Q \cup U^*$ \\
        Remove $U = U \setminus U^*$, update $p = p + p^*$.\\
    }\label{restrictOptions:endWhileQB}
    Set $\pv_Q = \sum_{(i,j) \in Q} p_{ij}\pv_{ij}$ and $\pv_B = p^* \pv^*$ \label{restrictOptions:computeValues}\\
    \textbf{if} $\pv_Q < \pv_B$ \textbf{then} set $Q = B$ \label{restrictOptions:chooseBetterSet} \\
    \textbf{return} $Q$
\end{algorithm}

If the set $Q$ returned by RestrictOptions only spans a single round $i$, the agent will always be incentivized to propose an acceptable option in round $i$. For this scenario, the algorithm creates a single bin $B_1$. Otherwise, the set $Q$ is divided into $c = \lceil \log_2 \alpha \rceil$ classes depending on their agent utility in line~\ref{algo-loglog-log:constructClasses}.
The lowest and highest agent utilities in any given class differ by at most a factor of 2. More precisely, classes $\mathcal{C}_1, \dots, \mathcal{C}_c$ are constructed such that $\mathcal{C}_k = \{(i,j) \in Q \mid a_{ij} \in [2^{k-1},2^k)\}$ for $k=1,\dots,c-1$ and $\mathcal{C}_c = \{(i,j) \in Q \mid a_{ij} \in [2^{c-1},2^c]\}$.

These classes are then (lines~\ref{algo-loglog-log:startFor}-\ref{algo-loglog-log:endFor}) combined into bins such that (1) the bins are as big as possible and (2) $\agent$ optimizes his own expected utility by proposing the first option he encounters from any bin -- assuming that only options from this bin are allowed. Classes are either added to a bin completely or not at all. Let $s$ be the index of the class with the highest agent utilities currently considered, i.e., the first class to be added to the current bin $B_b$. We consider the classes by decreasing agent utility values, i.e., with indices $k=s,s-1, \dots$. While $2^{k-1} \ge 2^s \cdot \sum_{(i,j) \in B_b \cup \mathcal{C}_k} p_{ij}$, a rational $\agent$ will always propose the first option available from the current bin if that is the only allowed bin as it has a higher utility than what $\agent$ can expect from later rounds.

Before stating the main result of the section, we analyze the subroutine RestrictOptions.

\begin{lemma}\label{lem:restrictOptions}
    The subroutine RestrictOptions$(\dist_1,\dots,\dist_n,m)$ with distributions $\dist_1, \dots, \dist_n$ and a parameter $0 < m \le 1$ as input identifies $Q$, the best set of options for $\princ$, such that
    \[\sum_{(i,j) \in Q} p_{ij}\pv_{ij} \ge m/2 \cdot \expec_{\omega_{ij}\sim \dist_i}[\max_{i \in [n]} \pv_{ij}] \]
    and either (1) the combined mass $\sum_{(i,j) \in Q} p_{ij} < m$ or (2) all options in $Q$ arrive in the same round.
\end{lemma}
\begin{proof}
    RestrictOptions first restricts the possible set of options to $Q \cup B$ consisting of the best options for $\princ$, with a union probability mass of at least $m$ in the while-loop in lines~\ref{restrictOptions:startWhileQB}-\ref{restrictOptions:endWhileQB}.
    Inside this while-loop, the options with the highest utility for $\princ$ are identified using the for-loop in lines~\ref{restrictOptions:startForSamePrincValue}-\ref{restrictOptions:endForSamePrincValue}. This loop ensures that no more than a combined mass of $m$ is considered for a set of options from different rounds with the same (currently highest) utility for $\princ$. Should such a set with a higher combined mass than $m$ exist, the if/else-statement in lines~\ref{restrictOptions:startIfCombinedMass}-\ref{restrictOptions:endIfCombinedMass} picks the better part of this set while ensuring that either the combined mass of the set is at most $m$ or only options from a single round are considered.

    Hence, by line~\ref{restrictOptions:endWhileQB} it holds that
    \[ 1/m \cdot \sum_{(i,j) \in Q \cup B} p_{ij}\pv_{ij} \ge  \expec_{\omega_{ij} \sim \dist_i}[\max_{i \in [n]} \pv_{ij}] \enspace.\]
    In line~\ref{restrictOptions:computeValues}, the utility for $\princ$ from the sets $Q$ and $B$ is calculated, in line~\ref{restrictOptions:chooseBetterSet}, the better one for $\princ$ is chosen. This means that at most another factor of 2 is lost, which means that in total, the set $Q$ which is returned by RestrictOptions guarantees
    \[  2/m \cdot \sum_{(i,j) \in Q} p_{ij}\pv_{ij} \ge \expec_{\omega_{ij} \sim \dist_i}[\max_{i \in [n]} \pv_{ij}] \enspace.\]
\end{proof}

Theorem~\ref{theo:loglog-log} is our main result in this section. We prove it below using Lemmas~\ref{lem:OneOverClasses} and~\ref{lem:numberOfBins}.

\begin{theorem}\label{theo:loglog-log}
    If the agent has $\alpha$-bounded utilities, there is a deterministic action scheme such that $\princ$ obtains an $\Omega(\log \log \alpha/\log \alpha)$-approximation of the expected utility for optimal (online) search.
\end{theorem}

\begin{lemma}
	\label{lem:OneOverClasses}
	Let $\ell$ be the number of bins opened in Algorithm~\ref{algo:log-alpha-loglog-alpha-approx}. Then the scheme computed by the algorithm obtains at least an $1/(8\ell)$-approximation of the expected utility of the best option for $\princ$ in hindsight.
\end{lemma}

\begin{proof}
    By Lemma~\ref{lem:restrictOptions}, we know that $Q$ satisfies
    \[
   		4 \cdot \sum_{(i,j) \in Q } p_{ij} \pv_{ij} \ge \expec_{\omega_{ij} \sim \dist_i} [\max_{i \in [n]} \pv_{ij} ] = \opt\enspace.
    \]
   	Now consider the construction of the bins. 
    Suppose 
    we split $Q$ into $\ell$ bins $B_1, B_2, \ldots, B_{\ell}$. In the end, we choose the best one $B_{b^*}$ among the $\ell$ bins $B_1,\ldots,B_{\ell}$, so
    \[
   	\sum_{(i,j) \in B_{b^*}} p_{ij} \pv_{ij} \ge \frac{1}{\ell} \sum_{(i,j) \in Q} p_{ij} \pv_{ij} \ge  \frac{1}{4\ell} \cdot \opt\enspace.
   	\]

   	The action scheme restricts attention to $B_{b^*}$ and accepts each proposed option $\omega_{ij}$ from the bin with probability 1. Let $k^- = \min \{ k \mid \mathcal{C}_k \subseteq B_{b^*}\}$ be the class of smallest index in $B_{b^*}$, and $k^+$ the one with largest index, respectively. Now suppose the agent learns in round $i$ that an option $\omega_{ij}$ with $(i,j) \in B_{b^*}$ arrives in this round. We claim that $\agent$ will then decide to propose this option. This is obvious if all options in $B_{b^*}$ are only realized in round $i$. Otherwise, the agent might want to wait for an option in a later round.
    If $\agent$ proposes, then his utility is $\av_{ij}$. Otherwise, if he waits for another option from $B_{b^*}$ in a later round, then a union bound shows that the expected utility is at most
   	\begin{align*}
    	\sum_{(i',j') \in B_{b^*}, i' > i} p_{i'j'} \cdot \av_{i'j'} \le \sum_{(i',j') \in B_{b^*}, i' > i} p_{i'j'} 2^{k^+} < 2^{k^+} \cdot \sum_{(i',j') \in B_{b^*}} p_{i'j'} \le  2^{k^- - 1}  \le \av_{ij} \enspace,
   	\end{align*}
   	where the second-to-last inequality is a consequence from the construction of the bin. Hence, the first option from the bin that is realized also gets proposed by $\agent$ and accepted by $\princ$.

   	Now for each option $(i,j) \in B_{b^*}$, the probability that this option is proposed and accepted is the combination of two
 	independent events: (1) no other option from $B_{b^*}$ was realized in any of the rounds $i' < i$, (2) option $\omega_{ij}$ is realized in round $i$. The probability for event (2) is $p_{ij}$. For event (1), we define  $m_i = \sum_{(i,j) \in B_{b^*}} p_{ij}$. With probability $\prod_{i' < i} (1-m_i')$, no option from $B_{b^*}$ is realized in rounds $i' < i$. Note that $\sum_{i=1}^n m_i \le 1/2$. The term $\prod_{i=1}^n (1-m_i)$ is minimized for $m_1 = 1/2$ and $m_{i'} = 0$ for $1 < i' < i$. Thus $\prod_{i=1}^n (1-m_i) \ge 1/2$, i.e., the probability of event (1) is at least 1/2.

   	Overall, by linearity of expectation, the expected utility of $\princ$ when using $\varphi$ is at least
   	\[
   		\sum_{(i,j) \in B_{b^*}} \frac 12 \cdot p_{ij} \cdot \pv_{ij} \ge \frac{1}{2} \sum_{(i,j) \in B_{b^*}} p_{ij} \pv_{ij} \ge \frac{1}{8\ell} \cdot \opt
   	\]
   	and this proves the lemma.
\end{proof}

\begin{lemma}\label{lem:numberOfBins}
	Let $\ell$ be the number of bins opened in Algorithm~\ref{algo:log-alpha-loglog-alpha-approx}. It holds that $\ell = O(\log \alpha / \log \log \alpha)$.
\end{lemma}

\begin{proof}
    Consider some bin $B$ and the mass $p_B = \sum_{(i,j) \in B} p_{ij}$. We want to argue that at most $O(c/\log c)$ bins are opened. To do so, we first condition on having $\ell$ open bins and strive to lower bound the number of classes in these $\ell$ bins.

    Consider a bin $B$ starting at $\mathcal{C}_s$. The algorithm adds classes to $B$ until $2^{k-1} < 2^s p_{B}$. Thus, $s-k+1 > \log_2(1/p_{B})$, i.e., the number of classes in $B_i$ is lower bounded by $\log_2(1/p_B)$.

    Now consider two bins $B_i$ and $B_{i+1}$ and condition on $q = p_{B_i} + p_{B_{i+1}}$. Together the bins contain at least $\log_2(1/p_{B_i}) + \log_2(1/(q-p_{B_i}))$ classes. Taking the derivative for $p_{B_i}$, we see that this lower bound is minimized when $p_{B_i} = q/2 = p_{B_{i+1}}$. Applying this balancing step repeatedly, the lower bound on the number of classes in all bins becomes smallest when $p_{B_i} = p_{B_j}$ for all bins $B_i, B_j$. Thus, when opening $\ell$ bins, we obtain the smallest lower bound on the number of classes in these bins by setting $p_{B_i} = \frac{1}{\ell} \cdot \sum_{(i,j) \in Q} p_{ij} < 1/(2\ell)$ for all bins $B_i$. Conversely, when opening $\ell$ bins, we need to have at least $\ell \log_2 (2\ell)$ classes in these bins.

    Now, since we need to put $c$ classes into the bins, we need to ensure that for the number $\ell$ of open bins we have
    $\ell (\log_2 \ell + 1) \le c$, since otherwise the $\ell$ bins would require more than $c$ classes in total. This implies that $c = \Omega( \ell \log_2 \ell )$ and, hence, $\ell = O(c / \log c) = O(\log \alpha / \log \log \alpha)$.
\end{proof}

Observe that the approximation ratio of this algorithm is tight in general. Consider the instances in Theorem~\ref{thm:generalLB} with $\alpha = n^{O(n)}$. The theorem shows that every scheme can obtain at most a ratio of $O(1/n) = O(\log \log \alpha / \log \alpha)$. 


\begin{remark}\rm\label{rem:loglogWithZero}
If there are options with utility 0 for $\agent$, the maximum ratio between the lowest and highest utility for $\agent$ becomes unbounded. Still, if the maximum ratio between the lowest and highest \emph{non-zero} utility for $\agent$ is bounded by $\alpha$, an $\Omega(\log \log \alpha / \log \alpha)$-approximation can be achieved with a slight modification of Algorithm~\ref{algo:log-alpha-loglog-alpha-approx}. Suppose there are any options with $\av_{ij} = 0$ in $Q$, then construct another bin $B_{-1}$ which consists of all options with 0 utility for $\agent$ in the set $Q$. If $B_{-1}$ is the bin that is chosen as the best bin in the algorithm, the agent will not receive any utility and, due to tie-breaking in favor of $\princ$, can be assumed to execute an online search for $\princ$. Using standard prophet inequality results, this yields a $1/2$-approximation for $\princ$ within this bin. If bin $B_{-1}$ is not the best bin, the analysis from the theorem can be applied.
\end{remark}

\begin{corollary}
	\label{cor:loglogWithZero}
	If the agent has $\alpha$-bounded utilities for all options with strictly positive utility,
	 there is a deterministic action scheme such that $\princ$ obtains an $\Omega(\log \log \alpha/\log \alpha)$-approximation of the expected utility for optimal (online) search.
\end{corollary}

\subsection{Oblivious Proposals}

In the previous section, we considered algorithms for $\princ$ when she learns the utility pair for the proposed option. In this section, we show that (fully) oblivious proposals can be a substantial drawback for $\princ$. Obviously, the lower bound in Theorem~\ref{thm:generalLB} remains intact even for oblivious proposals, when $\princ$ does not learn the utility value of the proposed option for $\agent$. For oblivious proposals and $\alpha$-bounded agent utilities, we can significantly strengthen the lower bound. In contrast to the logarithmic approximation guarantee above, we provide a linear lower bound in $\alpha$ for oblivious proposals.

\begin{theorem}
	\label{thm:fullyObliviousLB}
	There is a class of instances of online delegation with $\alpha$-bounded utilities for the agent and oblivious proposals, in which every action scheme $\varphi$ obtains at most an $O(1/\alpha)$-approximation of the expected utility for optimal (online) search.
\end{theorem}

\begin{proof}
	Consider the following class of instances. In $\dist_i$, there are two options with the following probabilities and utilities: $\omega_{i1}$ with $p_{i1} = 1-1/n$ and $(\pv_{i1},\av_{i1}) = (0,1)$, as well as $\omega_{i2}$ with $p_{i2} = 1/n$ and $(\pv_{i2},\av_{i2}) = (1,x_i)$, where $x_i \in \{1,\alpha\}$ and $\alpha \in [1,n]$. In the first rounds $i=1,\ldots,i^*-1$ we have $x_i = 1$, then $x_i = \alpha$ for rounds $i=i^*,\ldots,n$. The expected utility when $\princ$ performs (undelegated) online search is $1 - (1-1/n)^n \ge 1-1/e$.

	Clearly, $\princ$ has an incentive that $\agent$ proposes any profitable option $\omega_{i2}$ as soon as possible. As in the proof of Theorem~\ref{thm:generalLB}, we can assume that all $\varphi_{i1} = 0$ in an optimal scheme -- this option yields no value for $\princ$ and could only raise the incentive to wait for $\agent$.

	Due to oblivious proposals, $\princ$ has to choose $\varphi$ without being aware of the value of $i^*$. For our impossibility result, we adjust $i^*$ to the scheme $\varphi$ chosen by $\princ$: Set $i^* \in \{1,\ldots,n\}$ to the largest number such that $\sum_{i=i^*}^n \varphi_{i2} \ge e \cdot n/\alpha$, or $i^*=1$ if no such number exists.

	First, suppose that $i^* = 1$. Then, even if we make $\agent$ propose \emph{every} option $\omega_{i2}$ as soon as it arises, a union bound implies that the expected utility of $\princ$ is upper bounded by $\sum_{i=1}^n \frac{1}{n} \cdot \varphi_{i2} \le \frac{e}{\alpha} + \frac{1}{n}$. Hence, $\princ$ obtains only an $O(1/\alpha)$-approximation, for any $\alpha \in [1,n]$.

	Now suppose that $i^* > 1$.	Consider an optimal scheme $\varphi$ for $\princ$. If $\omega_{i2}$ arises in round $i$, $\agent$ decides if it is more profitable to propose $i$ or wait for a later round. Indeed, we show that $\agent$ never proposes $\omega_{i2}$ in any round $i < i^*$. Consider the expected utility from proposing the first option $\omega_{k2}$ arising in rounds $k = i^*,\ldots,n$. This is
	\begin{align*}
		\alpha \cdot\left( \sum_{k=i^*}^n \frac{1}{n} \left(1-\frac{1}{n}\right)^{k-i^*} \varphi_{k2} \right)& = \alpha \cdot \frac{1}{n} \cdot \sum_{k=i^*}^n\left(1-\frac{1}{n}\right)^{k-i^*}\varphi_{k2} \\
		&> \frac{\alpha}{n} \cdot \frac{1}{e} \sum_{k=i^*}^n \varphi_{k2} \ge \frac{\alpha}{en} \cdot \frac{en}{\alpha} = 1 \ge \varphi_{i2}\enspace,
	\end{align*}
	i.e., strictly larger than the expected utility $\varphi_{i2}$ from proposing $\omega_{i2}$ in round $i < i^*$. Hence, $\agent$ only proposes in rounds $k = i^*,\ldots,n$. Now even if $\agent$ would be able to propose \emph{every} option $\omega_{k2}$ in rounds $k=i^*,\ldots,n$, a union bound implies that the expected utility of $\princ$ from these rounds is upper bounded by $\sum_{k=i^*}^n \frac{1}{n} \cdot \varphi_{k2} \le \frac{e}{\alpha} + \frac{1}{n}$. For any $\alpha \in [1,n]$, this implies $\princ$ obtains an $O(1/\alpha)$-approximation.
\end{proof}

\begin{algorithm}[t]\DontPrintSemicolon
    \caption{$\Omega(1/\alpha)$-Approximation for Oblivious Proposals}\label{algo:alpha-approx}
    \KwIn{$n$ distributions $\dist_1, \dots, \dist_n$}
    \KwOut{Action Scheme $\varphi$}
    Let $Q = $ RestrictOptions$(\dist_1,\dots,\dist_n,1/(2\alpha))$. \\
    Set $\varphi_{ij} = 1$ for all $(i,j) \in Q$. \\
    \Return $\varphi$
\end{algorithm}
\begin{theorem}
    \label{thm:fullyObliviousAlgo}
    If the agent has $\alpha$-bounded utilities and makes oblivious proposals, there is a deterministic action scheme such that $\princ$ obtains an $\Omega(1/\alpha)$-approximation of the expected utility for optimal (online) search.
\end{theorem}

\begin{proof}
    The proof follows along the lines of Lemma~\ref{lem:OneOverClasses} above. By Lemma~\ref{lem:restrictOptions}, we have $4 \alpha \cdot \sum_{(i,j) \in Q} p_{ij} \pv_{ij} \ge \opt$,
    the expected value of the best option in hindsight.

    The action scheme accepts each proposed option $\omega_{ij}$ from the set $Q$ with probability 1. Note that $Q$ satisfies either that $\sum_{(i,j) \in Q} p_{ij} < 1/(2\alpha)$ or all options in $Q$ arrive in the same round $i$.

    In the latter case, $\agent$ will propose any option $\omega_{ij}$ with $(i,j) \in Q$ he encounters in round $i$. In a later round $i' > i$, $\princ$ will not accept any option.

    Hence, let us consider the former case that $Q$ satisfies $\sum_{(i,j) \in Q} p_{ij} < 1/(2\alpha)$. Suppose the agent learns in round $i$ that an option $\omega_{ij}$ with $(i,j) \in Q$ arrives. We claim that $\agent$ will propose this option. If $\agent$ proposes, then the expected utility is $\av_{ij}$. Otherwise, if he waits for another option from $Q$ in a later round, then a union bound shows that the expected utility is at most
    \begin{align*}
        &\sum_{(i',j') \in Q,\ i' > i} p_{i'j'} \cdot \av_{i'j'} \le \sum_{(i',j') \in Q,\ i' > i} p_{i'j'} \cdot \alpha \cdot \av_{ij} \le \av_{ij}\enspace,
    \end{align*}
    where the first inequality is due to $\alpha$-bounded utilities, and the second inequality follows since $\sum_{(i,j) \in Q} p_{ij} \le 1/(2\alpha)$ by construction. Hence, the first option from $Q$ that is realized also gets proposed by $\agent$ and accepted by $\princ$.

    Now, for each option $(i,j) \in Q$, the probability that this option is proposed and accepted is the combination of two independent events: (1) no other option from $Q$ was realized in any of the rounds $i' < i$, (2) option $\omega_{ij}$ is realized in round $i$. The probability for event (2) is $p_{ij}$. For probability for event (1) we define $m_i = \sum_{(i,j) \in Q} p_{ij}$. With probability $\prod_{i' < i} (1-m_i)$ no option from $Q$ is realized in rounds $i' < i$. Note that $\sum_{i=1}^n m_i \le 1/(2\alpha)$. The term $\prod_{i=1}^n (1-m_i)$ is minimized for $m_1 = 1/(2\alpha)$ and $m_{i'} = 0$ for $1 < i' < i$. Thus $\prod_{i=1}^n (1-m_i) \ge 1- 1/(2\alpha)$, i.e., the probability of event (1) is at least $1-1/(2\alpha) \ge 1/2$.

    By linearity of expectation, the expected utility of $\princ$ when using $\varphi$ based on $Q$ is at least
    \[
    \sum_{(i,j) \in Q} \frac{1}{2} \cdot p_{ij} \cdot \pv_{ij} \ge \frac{1}{8\alpha} \opt \enspace.
    \]
\end{proof}

In contrast to Corollary~\ref{cor:loglogWithZero}, the result of Theorem~\ref{thm:fullyObliviousAlgo} does \emph{not} generalize to the case when $\agent$ has options with utility 0, and $\alpha$ is the ratio of maximum and minimum non-zero utility. Even in the semi-oblivious scenario (discussed in the next section), all algorithms must have a ratio in $O(1/n)$, even when all utilities for $\agent$ are $a_{ij} \in \{0,1\}$.

\subsection{Semi-Oblivious Proposals}

In this section, we analyze semi-oblivious proposals, where $\princ$ has full apriori information about the prior, but she does not learn the utility value of $\agent$ upon a proposal. The additional information about the prior can indeed help to improve the approximation ratio from $\Theta(1/\alpha)$ to $\Omega(1/(\sqrt{\alpha} \log \alpha))$, but not to a logarithmic bound as shown for conscious proposals in Theorem~\ref{theo:loglog-log}. In particular, we start by showing the following limit on the approximation ratio.

\begin{theorem}\label{thm:semiObliviousLB}
	There is a class of instances of online delegation with IID options, $\alpha$-bounded utilities for the agent, and semi-oblivious proposals, in which every action scheme $\varphi$ obtains at most an $O(1/\sqrt{\alpha})$-approximation of the expected utility for optimal (online) search.
\end{theorem}

An IID instance with three different options suffices. One option is bad for both $\princ$ and $\agent$, but has a very high probability of $1-1/n$. The remaining two options provide the same (good) utility for $\princ$, one of which is good and the other one bad for $\agent$. The combined probability of both options is $1/n$. Since $\princ$ cannot distinguish between the two good options, in each of the rounds she has to decide to either accept both or reject both. While $\princ$ would like to accept any of the good options, $\agent$ has an incentive to wait and propose only the option that is good for both. Overall, this implies that every achievable approximation ratio for $\princ$ must be in $O(1/\sqrt{\alpha})$.

\begin{proof}[Proof of Theorem~\ref{thm:semiObliviousLB}]
    Consider the following class of IID instances with $\dist_i = \dist_j = \dist$. In $\dist$, there are three options with the following probabilities and utilities: $\omega_1$ with $p_1 = 1-1/n$ and $(\pv_1,\av_1) = (0,1)$, $\omega_2$ with $p_2 = 1/n - 1/(n\sqrt{\alpha})$ and $(\pv_2,\av_2) = (1,2)$, and $\omega_3$ with $p_3 = 1/(n\sqrt{\alpha})$ and $(\pv_3,\av_3)= (1,\alpha)$, for any $\alpha \in [2,n^2]$.

    Note that $\princ$ cannot distinguish between the latter options when they are proposed. Thus, in each round $i$, $\princ$ accepts option $\omega_1$ with probability $\varphi_{i1}$ and options $\omega_2$ and $\omega_3$ with $\varphi_{i2}$. As in the proof of Theorem~\ref{thm:generalLB}, we can assume that all $\varphi_{i1} = 0$ in an optimal scheme -- this option yields no value for $\princ$ and could only raise the incentive to wait for $\agent$.

    Consider any optimal scheme $\varphi$ for $\princ$. To obtain an upper bound on the utility of $\princ$, we assume that $\agent$ always proposes $\omega_3$ whenever it is realized\footnote{Note that due to the differences in acceptance probabilities $\varphi_{i2}$, he might actually have an incentive to wait for a later round, in which the probability that $\princ$ accepts is higher.}. For $\omega_2$, he evaluates whether or not it is profitable to wait for a later round. Suppose $\agent$ proposes $\omega_2$ in round $i$. A necessary condition for this is that the expected utility from proposing $\omega_3$ in subsequent rounds is smaller, i.e.,
    \begin{equation}\label{eq:necessaryConditionAgentProposal}
        2 \cdot \varphi_{i2} \ge \alpha \cdot\left( \sum_{k=i+1}^n \frac{1}{n\sqrt{\alpha}} \left(1-\frac{1}{n\sqrt{\alpha}}\right)^{k-i-1} \varphi_{k2} \right) =  \frac{\sqrt{\alpha}}{n} \cdot \sum_{k=i+1}^n\left(1-\frac{1}{n\sqrt{\alpha}}\right)^{k-i-1}\varphi_{k2} \enspace.
    \end{equation}
    If this condition is fulfilled, we set $\delta_i = 1$. Otherwise, we set $\delta_i = 0$. Then, using a union bound, the utility of $\princ$ from $\varphi$ can be upper bounded by
    \begin{equation}
        \label{eq:LBprincUB}
        \sum_{i = 1}^n \varphi_{i2} \left(\frac{1}{n\sqrt{\alpha}} + \delta_i \left(\frac{1}{n}-\frac{1}{n\sqrt{\alpha}}\right) \right) \enspace.
    \end{equation}
    Consider the first round $i_s$ in which $\delta_{i_s} = 1$.
    Combining \eqref{eq:necessaryConditionAgentProposal} with the fact that $\varphi_{i_s,2} \le 1$, this means that
    \[
    2 \ge \frac{\sqrt{\alpha}}{n} \cdot \sum_{k=i_s + 1}^n\left(1-\frac{1}{n\sqrt{\alpha}}\right)^{k-i_s-1}\varphi_{k2} \enspace,
    \]
    which implies
    \[
    \sum_{k=i_s+1}^n \varphi_{k2} < \frac{2n}{\sqrt{\alpha}} \cdot \left(1-\frac{1}{n\sqrt{\alpha}}\right)^{-n} < \frac{n}{\sqrt{\alpha}} \cdot 2e^{1/\alpha} \enspace.
    \]
    Using \eqref{eq:LBprincUB} and our assumption that $\alpha \in [2,n^2]$, the utility of $\princ$ is upper bounded by
    \[
    \frac{1}{n} \sum_{i = 1}^n \varphi_{i2} \left(\frac{1}{\sqrt{\alpha}} + \delta_i \left(1-\frac{1}{\sqrt{\alpha}}\right) \right) \le \frac{1}{n} \left(\frac{i_s - 1}{\sqrt{\alpha}} + 1 + 2e^{1/\alpha} \frac{n}{\sqrt{\alpha}} \right) = O(1/\sqrt{\alpha}) \enspace.
    \]
\end{proof}

For semi-oblivious proposals we design a more elaborate algorithm. The resulting action scheme provides an $\Omega(1/(\sqrt{\alpha} \log \alpha))$-appro\-xi\-mation for $\princ$. 
Our algorithm uses two subroutines, depending on the expected utility for $\agent$ (for pseudocode see Algorithm~\ref{algo:semiOblivious-approx}).

\begin{algorithm}[t]\DontPrintSemicolon
    \caption{$\Omega(1/(\sqrt{\alpha}\log \alpha))$-Approximation for Semi-Oblivious Proposals}\label{algo:semiOblivious-approx}
    \KwIn{$n$ distributions $\dist_1, \dots, \dist_n$}
    \KwOut{Action Scheme $\varphi$}
    Set $U = \bigcup_{i=1}^n \bigcup_{j=1}^{s_i} \{ (i,j) \}$.\\
    Partition $U$ into $U_L = \{(i,j) \in U \mid \sum_{\stackrel{k=1}{b_{ik} = b_{ij}}}^{s_i} p_{ik}a_{ik} < \sqrt{\alpha} \sum_{\stackrel{k=1}{b_{ik} = b_{ij}}}^{s_i} p_{ik}\}$ and $U_H = U \setminus U_L$. \\
    \For{$k=1,\dots,n$}{
        $\dist_k^{(L)} \leftarrow \dist_k$, $\dist_k^{(H)} \leftarrow \dist_k$\;
        In $\dist_k^{(L)}$ set utilities of every option $(k, j) \in U_H$ to 0 for $\princ$ and 1 for $\agent$ \;
        In $\dist_k^{(H)}$ set utilities of every option $(k, j) \in U_L$ to 0 for $\princ$ and $\sqrt{\alpha}$ for $\agent$.
    }
    Set $\varphi_L = AlgoLow(\dist^{(L)}_1,\dots, \dist^{(L)}_n)$, \quad $\varphi_H = AlgoHigh(\dist^{(H)}_1, \dots, \dist^{(H)}_n)$. \\
    \textbf{return} $\varphi_L$ or $\varphi_H$ whichever yields better expected utility for $\princ$
    %
\end{algorithm}

Consider all options with the same utility for $\princ$ in a single round. This set of options has \emph{low agent expectation} if the conditional expected utility for $\agent$ in this set of options is less than $\sqrt{\alpha}$. Otherwise, it has \emph{high agent expectation}.
For the first subroutine, we concentrate on all options with low agent expectation. Hence, this subroutine is called AlgoLow (Algorithm~\ref{algo:semiOblivious-low-expec}).
\begin{algorithm}[t]\DontPrintSemicolon
    \caption{AlgoLow}\label{algo:semiOblivious-low-expec}
    \KwIn{$n$ distributions $\dist_1,\dots, \dist_n$, where in every distribution individually, the set of options with the same value for $\princ$ has an expectation for $\agent$ of less than $\sqrt{\alpha}$}
    \KwOut{Action Scheme $\varphi$}
    Set $Q = RestrictOptions(\dist_1, \dots, \dist_n, 1/2)$.\\
            %
    Set $\ell = 1$, $\pv_1 = p_1 = 0$, $\mathcal{C}_1 = \emptyset$ \\
    \For{ $k =1, \dots, n$}{\label{algoLow:startForClasses}
        Set $p^* = \sum_{(k,j) \in Q} p_{kj}$\\
        \textbf{if} $p_\ell + p^* > 1/\sqrt{\alpha}$ \textbf{then} set $\ell = \ell +1, \mathcal{C}_\ell = \{(k,j) \in Q\}, p_\ell = p^*$;\; \quad \textbf{else} add $\mathcal{C}_{\ell} = \mathcal{C}_{\ell} \cup \{(k,j) \in Q\}$
    }\label{algoLow:endForClasses}
    Set $\pv_\ell' = \sum_{(i,j)\in \mathcal{C}_{\ell'}} p_{ij}\pv_{ij}$ for all $1 \le \ell' \le \ell$.\\
    Choose $\ell^*$ such that $\pv_{\ell^*} \ge \pv_{\ell'}$ for all $1 \le \ell' \le \ell$. \\
    Set $\varphi_{ij} = 1$ for all $(i,j) \in \mathcal{C}_{\ell^*}$.\\
    \textbf{return} $\varphi$
\end{algorithm}

Other options are considered to receive a utility of 0 for $\princ$ and, thus, are excluded from consideration. The scheme $\varphi_L$ achieves an $\Omega(1/\sqrt{\alpha})$-approximation in the instance $\dist^{(L)}$, where only options with low agent expectation generate value for $\princ$.
Similarly, for options with high agent expectation we describe procedure AlgoHigh (Algorithm~\ref{algo:semiOblivious-high-expec}).
\begin{algorithm}[t]\DontPrintSemicolon
    \caption{AlgoHigh}\label{algo:semiOblivious-high-expec}
    \KwIn{$n$ distributions $\dist_1,\dots, \dist_n$, where in every distribution individually, the set of options with the same value for $\princ$ has an expectation for $\agent$ of at least $\sqrt{\alpha}$}
    \KwOut{Action Scheme $\varphi$}
    Set $Q = RestrictOptions(\dist_1, \dots, \dist_n, 1/4)$.\\
    \For{$k = 0, \dots, \lfloor \log_2 \sqrt{\alpha} \rfloor -1 $}{
        Set $\mathcal{C}_k = \{ (i,j) \in Q \mid \frac{\sum_{(i,j') \in Q, \pv_{ij} = \pv_{ij'}}p_{ij}a_{ij}}{\sum_{(i,j') \in Q, \pv_{ij} = \pv_{ij'}}p_{ij}} \in \left[\sqrt{\alpha} \cdot 2^k, \sqrt{\alpha} \cdot 2^{k+1}\right) \}$ \\
        Set $\pv_k = \sum_{(i,j) \in \mathcal{C}_k} p_{ij}\pv_{ij}$.\\
    }
    Set $\mathcal{C}_{\lfloor \log_2 \sqrt{\alpha} \rfloor} = \{ (i,j) \in Q \mid \frac{\sum_{(i,j') \in Q, \pv_{ij} = \pv_{ij'}}p_{ij}a_{ij}}{\sum_{(i,j') \in Q, \pv_{ij} = \pv_{ij'}}p_{ij}} \in \left[\sqrt{\alpha} \cdot 2^{\lfloor \log_2 \sqrt{\alpha}\rfloor}, \alpha\right] \}$ \\
    Set $\pv_k = \sum_{(i,j) \in \mathcal{C}_k} p_{ij}\pv_{ij}$.\\
    Choose $k$ such that $\pv_k \ge \pv_{k'}$ for all $k' = 0, \dots, \lfloor \log_2 \sqrt{\alpha} \rfloor$.
    Set $\varphi_{ij} = 1$ for all $(i,j) \in \mathcal{C}_k$.\\
    \textbf{return} $\varphi$
\end{algorithm}
The scheme $\varphi_H$ achieves an $\Omega(1/(\sqrt{\alpha}\log_2{\alpha}))$-approximation in the instance $\dist^{(H)}$, where only options with high agent expectation generate value for $\princ$. In the end, we choose the better scheme for $\princ$, thereby forfeiting at most another factor 2 of her optimal expected utility. Overall, our Algorithm obtains a ratio of $\Omega(1/(\sqrt{\alpha}\log \alpha))$.

\begin{theorem}
    \label{thm:semiObliviousAlgo}
    If the agent has $\alpha$-bounded utilities and makes semi-oblivious proposals, there is a deterministic action scheme such that $\princ$ obtains an $\Omega(1/ (\sqrt{\alpha}\log \alpha))$-approximation of the expected utility for optimal (online) search.
\end{theorem}

AlgoLow and AlgoHigh use the procedure RestrictOptions with a parameter $m = 1/2$ and $m=1/4$, respectively. For a formal description of the subroutine see Algorithm~\ref{algo:semiOblivious-restrict} above.

Let us give a brief intuition for AlgoLow. The algorithm leverages the low expectation for $\agent$ by restricting the number of rounds from which options are accepted. More precisely, it partitions the set $Q$ computed by RestrictOptions  into $O(\sqrt{\alpha})$ many classes according to contiguous time intervals of rounds. The action scheme $\varphi$ then accepts only options from the best class for $\princ$. The overall probability that any acceptable option arrives turns out to be high enough (to obtain a $O(1/\sqrt{\alpha})$-approximation for $\princ$) and low enough (such that $\agent$ wants to propose the first acceptable option rather than wait for a better one later on).

\begin{lemma}\label{lem:obliviousLowExpec}
    If the agent has $\alpha$-bounded utilities, makes semi-oblivious proposals, and all options have low agent expectation, AlgoLow (Algorithm~\ref{algo:semiOblivious-low-expec}) constructs a deterministic action scheme such that $\princ$ obtains an $\Omega(1/\sqrt{\alpha})$-approximation of the expected utility for optimal (online) search.
\end{lemma}
\begin{proof}

    The set $Q$ returned by $RestrictOptions(\dist_1,\dots,\dist_n,1/2)$ guarantees $4 \cdot \sum_{(i,j) \in Q} p_{ij}\pv_{ij} \ge \opt$ by Lemma~\ref{lem:restrictOptions}.

    When splitting the set $Q$ into classes in the beginning of the algorithm, it is guaranteed that no class spanning more than a single round has a combined probability mass greater than $1/\sqrt{\alpha}$.
    This means that whenever a new class is opened, the mass of the previous and the current one combined are greater than $1/\sqrt{\alpha}$. Hence, there are at most $2 \cdot \sqrt{\alpha}$ many classes in total.

    Now assume class $\mathcal{C}$ is chosen by the algorithm and some acceptable option arrives in round $i$. From the assumption that utilities are $\alpha$-bounded, we know that this option has an agent value of at least 1. By a union bound, the probability that any additional acceptable option from $\mathcal{C}$ arrives in a future round is at most $1/\sqrt{\alpha}$ (as all classes that consist of a higher mass than $1/\sqrt{\alpha}$ only have options from a single round). The conditional expectation for $\agent$ for any acceptable option in a future round is at most $\sqrt{\alpha}$. Hence, $\agent$ proposes the option in round $i$.

    Similar to Algorithm~\ref{algo:log-alpha-loglog-alpha-approx}, the probability that an action $(i,j)$ from the chosen class is proposed is the combination of two independent events: (1) no other option from this class was proposed in a prior round $i' < i$ and (2) $(i,j)$ is realized in round $i$. If the chosen class only consists of a single round, the probability for (1) is trivially 1, otherwise, we can use the same argumentation as in the proof of Lemma~\ref{lem:OneOverClasses} to see that the probability that round $i$ is reached is at least $1/2$. This means that $\princ$ achieves an expected utility of at least $1/2 \cdot \sum_{(i,j) \in \mathcal{C}} p_{ij}\pv_{ij}$.

    As there are at most $2\cdot \sqrt{\alpha}$ many classes
    and the algorithm chooses the best one for $\princ$, by running AlgoLow, she will achieve an expected utility of
    \[ \frac{1}{2} \cdot \sum_{(i,j) \in \mathcal{C}} p_{ij}\pv_{ij} \ge \frac{1}{4\sqrt{\alpha}} \cdot \sum_{(i,j) \in Q} p_{ij}\pv_{ij} \ge \frac{1}{16\sqrt{\alpha}} \opt = \Omega\left(\frac{1}{\sqrt{\alpha}}\right)\opt \enspace.
    \]

\end{proof}
AlgoLow classifies options only based on utility for $\princ$ and time intervals. AlgoHigh instead uses an approach similar to Algorithm~\ref{algo:log-alpha-loglog-alpha-approx}, namely classifying good options for $\princ$ by their utility for $\agent$. Since in the semi-oblivious scenario, options from a single round $i$ with the same utility for $\princ$ cannot be distinguished, the algorithm classifies options by their expected utility for $\agent$ such that the expectation for $\agent$ of all options in a single class differs by no more than a factor 2. Finally, the algorithm identifies the best one of these $O(\log \alpha)$ many classes.

\begin{lemma}\label{lem:semiObliviousHighExpec}
    If the agent has $\alpha$-bounded utilities, makes semi-oblivious proposals, and all options have high agent expectation, AlgoHigh (Algorithm~\ref{algo:semiOblivious-high-expec}) constructs a deterministic action scheme such that $\princ$ obtains an $\Omega(1/(\sqrt{\alpha}\log\alpha))$-approximation of the expected utility for optimal (online) search.
\end{lemma}
\begin{proof}
    Using $RestrictOptions(\dist_1,\dots,\dist_n,1/4)$, the algorithm first identifies the best options for $\princ$. By Lemma~\ref{lem:restrictOptions}, it holds that $8 \cdot \sum_{(i,j) \in Q} p_{ij}\pv_{ij} \ge \opt$.

    The set $Q$ is then further partitioned into $\lfloor \log_2 \sqrt{\alpha}\rfloor + 1$ smaller classes depending on their conditional expectation for $\agent$, namely such that the conditional expectation for $\agent$ of the elements in a class differs by at most a factor 2. Then, the class $\mathcal{C}$ such that $\sum_{(i,j)\in \mathcal{C}} p_{ij}\pv_{ij} \ge \sum_{(i,j) \in \mathcal{C}'} p_{ij}\pv_{ij}$ for all classes $\mathcal{C}'$ is chosen. This means that
    \[(\lfloor \log_2 \sqrt{\alpha}\rfloor + 1) \cdot \sum_{(i,j) \in \mathcal{C}} p_{ij}\pv_{ij} \ge \sum_{(i,j) \in Q} p_{ij}\pv_{ij} \ge 1/8 \cdot \opt \enspace.
    \]
    We denote by $E$ the lower bound for the expected $\agent$ utility of the interval of the chosen class $\mathcal{C}$. Recall that all utilities for $\agent$ are in the interval $[1,\alpha]$. This means that with a probability of at least $E/(2\alpha-E) \ge E/(2\alpha)$, a random element from $\mathcal{C}$ has an agent utility of at least $E/2$ -- otherwise, an expected utility of at least $E$ for $\agent$ would not be possible. Since the probability that another allowed option in a later round arrives is at most 1/4 due to the choice of $m=1/4$ for the call to $RestrictOptions$ and the expectation conditional on arrival of an allowed option is at most $2E$, $\agent$ always proposes the first option with a utility of at least $E/2$. This in turn means that the agent will propose the first element from $\mathcal{C}$ he encounters with a probability of at least $E/(2\alpha)$. Since $E \ge \sqrt{\alpha}$, the probability that $\agent$ proposes the first allowed element is at least $1/(2\sqrt{\alpha})$.

    In total, this means that $\princ$ achieves an expected utility of at least

    \[\frac{1}{2\sqrt{\alpha}} \cdot \sum_{(i,j) \in \mathcal{C}} p_{ij}\pv_{ij} \ge \frac{1}{2\sqrt{\alpha}} \cdot \frac{1}{8 \cdot (\lfloor \log_2 \sqrt{\alpha}\rfloor + 1)}\cdot \opt = \Omega\left(\frac{1}{\sqrt{\alpha}\log_2\alpha} \right) \opt \enspace. \]
\end{proof}

In contrast to Corollary~\ref{cor:loglogWithZero}, the result of Theorem~\ref{thm:semiObliviousAlgo} does \emph{not} generalize to the case when $\agent$ has options with utility 0, and $\alpha$ is the ratio of maximum and minimum non-zero utility. A simple adaptation of the proof of Theorem~\ref{thm:semiObliviousLB} shows that in this case all algorithms must have a ratio in $O(1/n)$, even when all utilities for $\agent$ are $a_{ij} \in \{0,1\}$.

We adapt the instance from the proof of Theorem~\ref{thm:semiObliviousLB} as follows. We set $p_1 = 1-1/n$ and $(b_1,a_1) = (0,0)$, $p_2 = 1/n-1/n^2$ and $(b_2, a_2) = (1,0)$, $p_3 = 1/n^2$ and $(b_3,a_3) = (1,1)$. Note that $\alpha = 1$ here, as there is only a single non-zero utility value for $\agent$.

Consider any deterministic scheme for $\princ$. Clearly, $\agent$ does not want to propose any option of value 0 for him until the last round in which options $p_2$ and $p_3$ are acceptable. By a union bound, the overall probability to propose an option of value 1 for $\princ$ is at most $(n-1) \cdot 1/n^2 + 1/n < 2/n$, so the expected utility of $\princ$ is in $O(1/n)$. By searching through the options herself, $\princ$ obtains a value of at least $1-(1-1/n)^n \ge 1-1/e$. Hence, every deterministic scheme is $O(1/n)$-approximate, even in this case with $\alpha = 1$. A similar argument shows this result also for randomized schemes.

	\section{Misalignment of Principal and Agent Utility}\label{sec:principal-agent-ratio}

    In this section, we consider performance guarantees based on the amount of misalignment of principal and agent utility. Formally, let $\beta \ge 1$ be the smallest number such that
    \[
    	\frac{1}{\beta} \cdot \frac{\av_{ij}}{\av_{i'j'}} \le \frac{\pv_{ij}}{\pv_{i'j'}} \le \beta \cdot \frac{\av_{ij}}{\av_{i'j'}}
    \]
    for any two options $\omega_{ij}$ and $\omega_{i'j'}$ in the instance. We assume that the preference of $\princ$ between any pair $\omega_{ij}, \omega_{i'j'}$ of options is shared by $\agent$ -- up to a factor of at most $\beta$. We term this $\beta$-bounded utilities.

    For most of the section, we assume that all utility values are strictly positive. Suppose we choose an arbitrary realization $\omega_{i'j'}$. Divide all utility values of $\princ$ for all realizations by $\pv_{i'j'}$, and all utility values of $\agent$ by $\av_{i'j'}$. Note that this adjustment neither affects the incentives of the players nor the approximation ratios of our algorithms. Considering $\omega_{ij}$ with the adjusted utilities, we see that $1/\beta \cdot \pv_{ij}/\av_{ij} \le 1 \le \beta \cdot \pv_{ij}/\av_{ij}$, and thus $1/\beta \le \pv_{ij}/\av_{ij} \le \beta$ for all $\omega_{i'j'}$. This condition turns out to be convenient for our analysis.

    Our main idea is to use $O(\log \beta)$ clusters $\mathcal{C}_k$ to group all the options that have a utility ratio between $2^k$ and $2^{k+1}$, i.e.,
    \[
    	\mathcal{C}_k = \{ \omega_{ij} \in \Omega \mid 2^k \le \pv_{ij}/\av_{ij} < 2^{k+1} \}
    \]
    for $k = \lfloor \log 1/\beta \rfloor, \ldots, \lceil \log \beta \rceil$. Our deterministic scheme restricts the acceptable options to a single cluster $\mathcal{C}_{k^*}$. Note that here $\princ$ is assumed to see $\av_{ij}$ upon a proposal. The principal determines the cluster $k^*$, such that the best response by $\agent$ (i.e., his optimal online algorithm applied with the options from that cluster) delivers the largest expected utility for $\princ$.

    \begin{theorem}
        If principal and agent have $\beta$-bounded utilities, there is a deterministic action scheme such that $\princ$ obtains an $\Omega(1/\log \beta)$-approximation of the expected utility for optimal (online) search.
    \end{theorem}

    \begin{proof}
		Consider any cluster $\mathcal{C}_k$. We denote by $\pv(\agent, k)$ and $\av(\agent, k)$ the expected utility for $\princ$ and $\agent$ when $\princ$ uses $\mathcal{C}_k$ to determine $\varphi$. Now consider a hypothetical algorithm for $\princ$ that observes all realizations and chooses the best option from $\mathcal{C}_k$ for $\princ$ if possible. If there is no such option, it obtains a utility of 0. Let $\pv(\princ, k)$ and $\av(\princ, k)$ be the expected utility of the hypothetical algorithm for $\princ$ and $\agent$, respectively. Clearly, $\pv(\princ, k) \ge \pv(\agent, k)$ and $\av(\agent, k) \ge \av(\princ, k)$, but also, by definition of $\mathcal{C}_k$,
		\[
			\pv(\agent, k) \ge \av(\agent, k) \cdot 2^k \ge \av(\princ, k) \cdot 2^k \ge \pv(\princ, k) / 2
		\]
		Now consider the best option for $\princ$ in hindsight. The best-option-algorithm for cluster $\mathcal{C}_k$ picks the best option in hindsight if it comes from cluster $\mathcal{C}_k$. Otherwise, it returns a value of 0. Let $\pv^*_k$ be the expected utility of this algorithm for $\princ$, and let $OPT$ be the expected utility of the best option in hindsight for $\princ$. Then
		\[
		OPT = \sum_{k= \lfloor \log 1/\beta \rfloor}^{ \lceil \log \beta \rceil} \pv^*_k \le  \sum_{k= \lfloor \log 1/\beta \rfloor}^{ \lceil \log \beta \rceil} \pv(\princ, k) \le \sum_{k= \lfloor \log 1/\beta \rfloor}^{ \lceil \log \beta \rceil}	\pv(\agent, k) \cdot 2 \enspace.
		\]
		Hence, since the scheme chooses the cluster $k^*$ that maximizes $\pv(\agent, k^*)$, we obtain an $\Omega(1/\log \beta)$-approximation.
	\end{proof}

By treating all options of utility 0 for $\agent$ in a separate class and ignoring all options of utility 0 for $\princ$, we can again adapt the performance guarantee also to instances, in which all utility pairs of $\agent$ and $\princ$ with strictly positive entries are $\beta$-bounded.

%
\begin{corollary}
	 If principal and agent have $\beta$-bounded utilities for the set of options with only strictly positive utilities, there is a deterministic action scheme such that $\princ$ obtains an $\Omega(1/\log \beta)$-approximation of the expected utility for optimal (online) search.
\end{corollary}

The bound in Theorem~\ref{thm:generalLB} for conscious proposals can be applied rather directly to this case, i.e., when treating the 0-utility options for $\princ$ in a separate class. Also the bounds for oblivious and semi-oblivious proposals in Theorems~\ref{thm:fullyObliviousLB} and~\ref{thm:semiObliviousLB} apply directly, since in these instances $\beta = \Theta(\alpha)$. This implies that any algorithm has a ratio in $O(\log \log \beta/ \log \beta)$ for conscious proposals, in $O(1/\sqrt{\beta})$ for semi-oblivious proposals, and in $O(1/\beta)$ for oblivious proposals. Finally, it is trivial to obtain a $\Omega(1/\beta)$-approximation for $\princ$ in case of $\beta$-bounded utilities and oblivious proposals -- simply accept every option proposed by $\agent$. The bound on the ratio is a simple consequence of $\beta$-boundedness. As such, note that $\princ$ is not required to know $\beta$ to obtain the approximation.

\bibliographystyle{plainurl}
\bibliography{literature,martin}


\end{document}